\documentclass[12pt]{article}

\usepackage{amssymb}
\usepackage{eucal}
\usepackage{amsmath}
\usepackage{amsthm}
\usepackage{graphicx}

\newcommand{\R}{{\mathbb  R}}

\numberwithin{equation}{section}
\newtheorem{thm}{\bf Theorem}[section]
\newtheorem{lem}[thm]{\bf Lemma}
\newtheorem{prop}[thm]{\bf Proposition}

\theoremstyle{remark}
\newtheorem{rem}{\bf Remark}[section]

\title{A note on stability of spacecrafts and underwater vehicles}
\author{Dan Com\u anescu\footnote{Corresponding author - E-mail: comanescu@math.uvt.ro; Phone: +40 256 592281; Fax: +40 256 592316}\\
{\small Department of Mathematics, West University of Timi\c soara}\\
{\small Bd. V. P\^ arvan, No 4, 300223 Timi\c soara, Rom\^ ania}\\
{\small E-mail addresses: comanescu@math.uvt.ro}}
\date{}

\begin{document}

\maketitle

\begin{abstract}
A Hamilton-Poisson system is an approach for the motion of a spacecraft around an asteroid or for the motion of an underwater vehicle.  We construct a coordinate chart on the symplectic leaf which contains a specific generic equilibrium point and we establish stability conditions for this equilibrium point.
\end{abstract}



\section{Introduction}

The problem of stability of spacecrafts and underwater vehicles attracted numerous resources and led to the emergence of a significant number of theoretical studies.

In this paper we consider an idealized dynamics of spacecrafts and underwater vehicles which is a Hamilton-Poisson system. For a spacecraft we consider the dynamics presented in \cite{wang-xu} and for an underwater vehicle we work with the dynamics considered in \cite{leonard-1996}. Specific Casimir functions of these situations allow description of the symplectic leaves containing some equilibrium points by using the rotation matrices. We construct a coordinate chart on a symplectic leaf specified above by means of the double covering map between the 3-sphere $S^3$ and the special orthogonal group $SO(3)$. The expression of Hamilton function by using this coordinate chart lead us to some stability results for particular states of spacecrafts and underwater vehicles.  To decide the stability we use the algebraic method presented in \cite{comanescu}. Alternative methods are Arnold method (see \cite{arnold}) or equivalent methods as Casimir method and Ortega-Ratiu method (see \cite{birtea-puta}). 

In the second section we present stability results of some generic equilibrium points of a spacecraft moving around an asteroid. We find the sufficient conditions for stability from the paper \cite{wang-xu}. Our method reduces to the study of the eigenvalues of a $6\times 6$ matrix while paper \cite{wang-xu} work with a $12\times 12$ matrix.

In Section 3 we study the stability of some generic equilibrium points of an underwater vehicle.
We consider a more general situation than studied in papers \cite{leonard-automatica}, \cite{leonard-1996}, \cite{leonard-1996-1}, \cite{leonard-marsden}; we suppose that the third axis is a principal axis of inertia for the vehicle but the first and second axes of the vehicle-fixed frame may not be principal axes of inertia for the vehicle.
We prove that the conditions for the stability of an equilibrium point does not depend on the position of the first and second axes of inertia of the vehicle in the perpendicular plane on the third axis of the vehicle-fixed frame.

\section{Stability of a spacecraft moving around an asteroid}

We consider a rigid spacecraft moving on a stationary orbit around a rigid asteroid. The fixed-body frame of the asteroid $(O,{\bf u},{\bf v},{\bf w})$ has the origin in the mass center and the axes are principal axes of inertia of the asteroid. According to the hypotheses made in the paper \cite{wang-xu}, we assume  that the mass center of the asteroid is stationary in an inertial frame, the asteroid has an uniform rotation around its maximum-moment principal axis, the spacecraft is on a stationary orbit, and the orbital motion is not affected by the attitude motion.

The fixed-body frame of the spacecraft  $(C,{\bf i},{\bf j},{\bf k})$ has the origin $C$ in the mass center and the axes are principal axes of inertia of the spacecraft. "A stationary orbit in the inertial frame corresponds to an equilibrium in the fixed-body frame of the asteroid; there are two kinds of stationary orbits: those that lie on the intermediate moment principal axis of the asteroid, and those that lie on the minimum-moment principal axis of the asteroid" (see \cite{wang-xu}).

The attitude of the spacecraft is desribed by the vectors $\boldsymbol{\alpha}, \boldsymbol{\beta}, \boldsymbol{\gamma}$, where $\boldsymbol{\gamma}$ is the versor with the origin in the mass center of the spacecraft towards the mass center of the asteroid, $\boldsymbol{\beta}$ is the versor in the opposite direction of the orbital angular momentul, and $\boldsymbol{\alpha}=\boldsymbol{\beta}\times \boldsymbol{\gamma}$. We denote by $\boldsymbol{\Pi}$ the angular momentum of the spacecraft with respect to the inertial frame and ${\bf z}=(\boldsymbol{\Pi},\boldsymbol{\alpha},\boldsymbol{\beta},\boldsymbol{\gamma})$.

\begin{figure}[h!]
  \caption{Spacecraft around an asteroid.}
  \centering
    \includegraphics[width=1.0\textwidth]{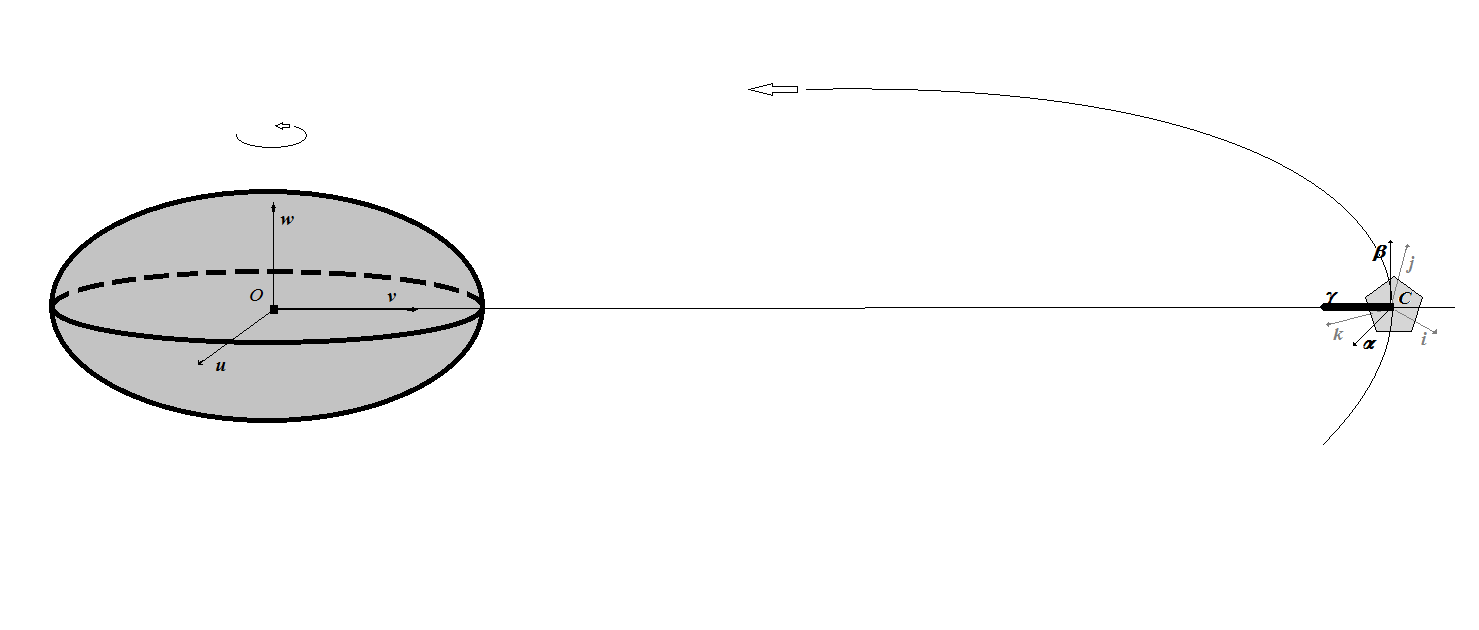}
\end{figure}

According to the paper \cite{wang-xu}, the system of motion can be written as
\begin{equation}\label{spacecraft-asteroid-eq}
\dot{\bf z}=\Lambda({\bf z})\nabla H({\bf z});
\end{equation}
The Hamiltonian function is given by
\begin{equation}
H({\bf z})=\frac{1}{2}<\boldsymbol{\Pi},\mathbb{I}^{-1}\boldsymbol{\Pi}>+\omega_T<\boldsymbol{\Pi},\boldsymbol{\beta}>+V(\boldsymbol{\alpha},\boldsymbol{\beta},\boldsymbol{\gamma}),\footnote{$<\cdot,\cdot>$ is the scalar product on $\R^3$.}
\end{equation}
where $\omega_T$ is the angular velocity of the uniform rotation of the asteroid, $\mathbb{I}$ is the inertia tensor of the spacecraft and its matrix in the fixed-body frame is $\text{diag}(I_1,I_2,I_3)$. The function $V$ represents the perturbation due to the gravity gradient torque and it has the expression:
$$V(\boldsymbol{\alpha},\boldsymbol{\beta},\boldsymbol{\gamma})=k_1<\boldsymbol{\alpha},\mathbb{I}\boldsymbol{\alpha}>+k_2<\boldsymbol{\beta},\mathbb{I}\boldsymbol{\beta}>+
k_3<\boldsymbol{\gamma},\mathbb{I}\boldsymbol{\gamma}>,$$
where $k_1,k_2$ and $k_3$ are constants.
The Poisson tensor  is given by
$$ \Lambda({\bf z})=\left(%
\begin{array}{cccc}
\widehat{\boldsymbol{\Pi}} & \widehat{\boldsymbol{\alpha}} & \widehat{\boldsymbol{\beta}} & \widehat{\boldsymbol{\gamma}} \\
 \widehat{\boldsymbol{\alpha}} & O_3 & O_3 & O_3 \\
\widehat{\boldsymbol{\beta}} & O_3 & O_3 & O_3 \\
\widehat{\boldsymbol{\gamma}} & O_3 & O_3 & O_3
\end{array}%
\right).
\footnote{$O_3$ is the null matrix, and for a vector ${\bf v}=(a,b,c)$ we denote by $\widehat{\bf v}:=\left(%
\begin{array}{ccc}
0 & -c & b \\
c & 0 & -a \\
-b & a & 0
\end{array}%
\right).$}
$$
The Casimir functions are $C_{11}({\bf z})=||\boldsymbol{\alpha}||^2$, $C_{12}({\bf z})=<\boldsymbol{\alpha},\boldsymbol{\beta}>$, $C_{13}({\bf z})=<\boldsymbol{\alpha},\boldsymbol{\gamma}>$, $C_{22}({\bf z})=||\boldsymbol{\beta}||^2$ where $1\leq i\leq j\leq 3$ and we denote by ${\bf C}:\R^{12}\rightarrow \R^6$ the vectorial Casimir function.
We have seven conserved quantities $H,C_{11},C_{12},C_{13},C_{22},C_{23},C_{33}$ of the dynamics generated by \eqref{spacecraft-asteroid-eq}. 

We observe that we have the following equilibrium point ${\bf z}^e$:
\begin{equation}\label{equilibrium-spacecraft}
{\boldsymbol{\Pi}}^e=-\omega_TI_2{\bf j},\,\,{\boldsymbol{\alpha}}^e={\bf i},\,\,{\boldsymbol{\beta}}^e={\bf j},\,\,{\boldsymbol{\gamma}}^e={\bf k}.
\end{equation}
According to \cite{birtea-comanescu}, it is a generic equilibrium point because it is a regular point of ${\bf C}$; i.e. $\text{rank} \,\nabla {\bf C}({\bf z}^e)=6$.

For our study we construct a coordinate chart on the symplectic leaf  ${\bf C}^{-1}({\bf C}({\bf z}^e))$ around the equilibrium point ${\bf z}^e$.
We consider the open subset:
$${\bf C}^+_{{\bf z}^e}=\{{\bf z}=(\boldsymbol{\Pi},\boldsymbol{\alpha},\boldsymbol{\beta},\boldsymbol{\gamma})\in {\bf C}^{-1}({\bf C}({\bf z}^e))\,|\,\text{sgn}<\boldsymbol{\alpha},\boldsymbol{\beta}\times\boldsymbol{\gamma}>=\text{sgn}<\boldsymbol{\alpha}^e,\boldsymbol{\beta}^e\times\boldsymbol{\gamma}^e>\}\footnote{$\text{sgn}$ is the signum function}$$
and we observe that ${\bf z}^e\in {\bf C}^+_{{\bf z}^e}$.

\begin{lem}
The function $\mathcal{F}_{{\bf z}_e}:SE(3)\rightarrow {\bf C}^+_{{\bf z}^e}$ given by
$\mathcal{F}_{{\bf z}^e}(\boldsymbol{\Pi},R):=(\boldsymbol{\Pi}, R\boldsymbol{\alpha}^e, R\boldsymbol{\beta}^e, R\boldsymbol{\gamma}^e)$
is a homeomorphism and we have $\mathcal{F}_{{\bf z}^e}(\boldsymbol{\Pi}^e,I_3)={\bf z}^e$.
\footnote{We use the special orthogonal group and the special Euclidean group,
$$SO(3)=\{X\in \mathcal{M}_3(\R)\,|\,XX^T=X^TX=I_3,\,\det(X)=1\},\,\,
SE(3):=\R^3\rtimes SO(3).$$
$\mathcal{M}_3(\R)$ is a normed space by using the Frobenius norm
$||X||_F=\sqrt{\text{Trace}(X^TX)}.$}

\end{lem}

\begin{proof}
If $R\in SO(3)$ and ${\bf u},{\bf v}\in \R^3$ then we have $<R{\bf u},R{\bf v}>=<{\bf u},{\bf v}>$ and consequently $(\boldsymbol{\Pi}, R\boldsymbol{\alpha}^e, R\boldsymbol{\beta}^e, R\boldsymbol{\gamma}^e)\in {\bf C}^{-1}({\bf C}({\bf z}^e))$. Also, we have
$$<R\boldsymbol{\alpha}^e,R\boldsymbol{\beta}^e\times R\boldsymbol{\gamma}^e>=<R\boldsymbol{\alpha}^e,R(\boldsymbol{\beta}^e\times \boldsymbol{\gamma}^e)>=<\boldsymbol{\alpha}^e,\boldsymbol{\beta}^e\times \boldsymbol{\gamma}^e>$$
which implies that $(\boldsymbol{\Pi}, R\boldsymbol{\alpha}^e, R\boldsymbol{\beta}^e, R\boldsymbol{\gamma}^e)\in{\bf C}^+_{{\bf z}^e}$.

Let $(\boldsymbol{\Pi},\boldsymbol{\alpha},\boldsymbol{\beta},\boldsymbol{\gamma})\in {\bf C}^+_{{\bf z}^e}$.  The sets $\{\boldsymbol{\alpha},\boldsymbol{\beta},\boldsymbol{\gamma}\}$ and $\{\boldsymbol{\alpha}^e,\boldsymbol{\beta}^e,\boldsymbol{\gamma}^e\}$ are orthonormal bases in $\R^3$ with the same orientation and consequently, there exists $R\in SO(3)$ such that $(\boldsymbol{\alpha},\boldsymbol{\beta},\boldsymbol{\gamma})=(R\boldsymbol{\alpha}^e,R\boldsymbol{\beta}^e,R\boldsymbol{\gamma}^e)$. We have the surjectivity.

If $\mathcal{F}_{{\bf z}^e}(\boldsymbol{\Pi},R)=\mathcal{F}_{{\bf z}^e}(\tilde{\boldsymbol{\Pi}},\tilde{R})$, then we have $\boldsymbol{\Pi}=\tilde{\boldsymbol{\Pi}}$, $R^T\tilde{R}\boldsymbol{\alpha}^e=\boldsymbol{\alpha}^e$, $R^T\tilde{R}\boldsymbol{\beta}^e=\boldsymbol{\beta}^e$, and $R^T\tilde{R}\boldsymbol{\gamma}^e=\boldsymbol{\gamma}^e$.
Because $\{\boldsymbol{\alpha}^e,\boldsymbol{\beta}^e,\boldsymbol{\gamma}^e\}$ is a base of $\R^3$ we have that $R^T\tilde{R}=\mathbb{I}_3$ and consequently $R=\tilde{R}$.  We have the injectivity.
\end{proof}

We construct a double covering map between $\R^3\times S^3$ and $SE(3)$, where $$S^3=\{{\bf q}:=(q_0,q_1,q_2,q_3)\in \R^4\,|\,q_0^2+q_1^2+q_2^2+q_3^2=1\}$$ is the 3-sphere.
For  ${\bf q}\in S^3$ we consider the rotation matrix:
 \begin{equation*}
\mathbf{R}^{\mathbf{q}}=\left(
\begin{array}{ccc}
q_0^{2}+q_{1}^{2}-q_{2}^{2}-q_{3}^{2} & 2(q_{1}q_{2}-q_{0}q_{3}) & 2(q_{1}q_{3}+q_{0}q_{2}) \\
2(q_{1}q_{2}+q_{0}q_{3})  & q_0^{2}-q_{1}^{2}+q_{2}^{2}-q_{3}^{2}& 2(q_{2}q_{3}-q_{0}q_{1}) \\
2(q_{1}q_{3}-q_{0}q_{2}) & 2(q_{2}q_{3}+q_{0}q_{1}) & q_0^{2}-q_{1}^{2}-q_{2}^{2}+q_{3}^{2}
\end{array}
\right).
\end{equation*}
The function  $\mathbf{P}:S^{3}\rightarrow SO(3)$, $\mathbf{P}(\mathbf{q})=R^{\mathbf{q}}$ is a smooth double covering map.\footnote{ For $R\in SO(3)$ there exists two distinct points ${\bf q}_1,{\bf q}_2\in S^3$ such that ${\bf P}({\bf q}_1)={\bf P}({\bf q}_2)=R$. ${\bf P}$ is a continuous surjective function with the following property: for all $R\in SO(3)$ there exists $U_R$ an open neighborhood of $R$ (evenly-covered neighborhood) such that ${\bf P}^{-1}(U_R)$ is a union of two disjoint open sets in $S^3$ (sheets over $U_R$), each of which is mapped homeomorphically onto $U_R$ by ${\bf P}$.} For ${\bf q}\in S^3$ we have $R^{\mathbf{q}}=R^{\mathbf{-q}}$ and $I_3=R^{(1,\mathbf{0})}=R^{(-1,\mathbf{0})}$.
We construct the double covering map $\widetilde{\bf P}:\R^3\times S^3\rightarrow SE(3)$ given by
$\widetilde{\bf P}(\boldsymbol{\Pi},{\bf q})=(\boldsymbol{\Pi},{\bf P}({\bf q})),$
which has the property $\widetilde{\bf P}(\boldsymbol{\Pi}^e,(1,{\bf 0}))=(\boldsymbol{\Pi}^e,I_3)$.

Between $\R^3\times B_3({\bf 0},1)\subset\R^6$ and $\R^3\times S^3_+:=\R^3\times \{{\bf q}=(q_0,q_1,q_2,q_3)\in S^3\,|\,q_0>0\}\subset \R^3\times S^3$ we have the following homeomorphism $\mathcal{G}(\boldsymbol{\Pi},{\bf p})=(\boldsymbol{\Pi},(\sqrt{1-||{\bf p}||^2},{\bf p}))$. It  has the property $\mathcal{G}(\boldsymbol{\Pi}^e,{\bf 0})=(\boldsymbol{\Pi}^e,(1,{\bf 0}))$.

The above considerations lead to the following result.

\begin{prop}\label{chart-spacecraft}
There exists an open subset $U\subset {\bf C}^+_{{\bf z}^e}\subset  {\bf C}^{-1}({\bf C}({\bf z}^e))$ such that ${\bf z}^e\in U$ and the set $U$ with the function $(\mathcal{F}_{{\bf z}^e}\circ \tilde{\bf P}\circ\mathcal{G})^{-1}:U\rightarrow V\subset \R^6$ is a coordinate chart.
\end{prop}
The Hamiltonian function written in the chart defined in Proposition \ref{chart-spacecraft} is:
$$H_{{\bf z}^e}(\boldsymbol{\Pi},{\bf p})=H(\boldsymbol{\Pi},R^{(\sqrt{1-||{\bf p}||^2},{\bf p})}\boldsymbol{\alpha}^e, R^{(\sqrt{1-||{\bf p}||^2},{\bf p})}\boldsymbol{\beta}^e, R^{(\sqrt{1-||{\bf p}||^2},{\bf p})}\boldsymbol{\gamma}^e).$$

We have the following stability result.

\begin{thm}\label{chart-spacecraft-Hess}
Let ${\bf z}^e$ be the equilibrium point described by \eqref{equilibrium-spacecraft}.
If the point $(\boldsymbol{\Pi}^e,{\bf 0})$ is  a strict local extremum of $H_{{\bf z}^e}$, then the equilibrium point ${\bf z}^e$ is Lyapunov stable.

\end{thm}

\begin{proof} We use the algebraic method, presented in the paper \cite{comanescu} and used in the papers \cite{comanescu-1}, \cite{comanescu-2} and \cite{birtea-casu}.
If $(\boldsymbol{\Pi}^e,{\bf 0})$ is  a strict local extremum of $H_{{\bf z}^e}$, then the algebraic system
\begin{equation}\label{algebraic-system}
H({\bf z})=H({\bf z}^e),\,C_{ij}({\bf z})=C_{ij}({\bf z}^e),\,\,1\leq i\leq j\leq 3
\end{equation}
 has no root besides ${\bf z}^e$ in some neighborhood of ${\bf z}^e$. Consequently, the equilibrium point is Lyapunov stable .
\end{proof}
An immediate consequence is the following result.
\begin{thm}\label{conditions-stability-Hessian}
Let ${\bf z}^e$ be the equilibrium point due by \eqref{equilibrium-spacecraft}.
If $(\boldsymbol{\Pi}^e,{\bf 0})$ is a stationary point for $H_{{\bf z}^e}$ and the Hessian matrix $\text{Hess}\,H_{{\bf z}^e}(\boldsymbol{\Pi}^e,{\bf 0})$ is positive or negative definite, then the equilibrium point ${\bf z}^e$ is Lyapunov stable. 
\end{thm}
We can announce the main result of this section.

\begin{thm}\label{stability-spacecraft} A sufficient condition for the Lyapunov stability of the equilibrium point \eqref{equilibrium-spacecraft} is:
\begin{align*}
& (I_3-I_1)(k_1-k_3)>0; \\
\text{and}\,\, &  (I_2-I_1)(\omega_T^2-2(k_2-k_1))>0; \\
\text{and}\,\, &  (I_2-I_3)(\omega_T^2-2(k_2-k_3))>0.
\end{align*}
\end{thm}

\begin{proof}
For the computations we use the components of the vectors in the fixed-body frame. We have
$$\boldsymbol{\alpha}=(1-2p_2^2-2p_3^2, 2p_1p_2+2p_3\sqrt{1-(p_1^2+p_2^2+p_3^2)}, 2p_1p_3-2p_2\sqrt{1-(p_1^2+p_2^2+p_3^2)}),$$
$$\boldsymbol{\beta}=(2p_1p_2-2p_3\sqrt{1-(p_1^2+p_2^2+p_3^2)}, 1-2p_1^2-2p_3^2, 2p_2p_3+2p_1\sqrt{1-(p_1^2+p_2^2+p_3^2)}),$$
$$\boldsymbol{\gamma}=(2p_1p_3+2p_2\sqrt{1-(p_1^2+p_2^2+p_3^2)}, 2p_2p_3-2p_1\sqrt{1-(p_1^2+p_2^2+p_3^2)}, 1-2p_1^2-2p_2^2),$$
and we observe that $(\boldsymbol{\Pi}^e,{\bf 0})$ is a critical point for $H_{{\bf z}^e}$. The Hessian matrix is:
$$\text{Hess}\,H_{{\bf z}^e}(\boldsymbol{\Pi}^e,{\bf 0})=
\left(%
\begin{array}{cccccc}
\frac{1}{I_1} & 0 & 0 & 0 & 0 & -2\omega_T \\
0 & \frac{1}{I_2} & 0 & 0 & 0 & 0 \\
0 & 0 & \frac{1}{I_3} & 2\omega_T & 0 & 0 \\
0 & 0 & 2\omega _T & h_{44} & 0 & 0 \\
 0 & 0 & 0 & 0 & 8(I_3-I_1)(k_1-k_3) & 0 \\
-2\omega_T &   0 & 0 & 0 & 0 &   h_{66}
\end{array}%
\right),$$
where
$h_{44}=4(\omega_T^2I_2+2(I_3-I_2)(k_2-k_3))$ and $h_{66}= 4(\omega_T^2I_2+2(I_1-I_2)(k_2-k_1)).$
If we have the above inequalities, then the Hessian matrix $\text{Hess}\,H_{{\bf z}^e}(\boldsymbol{\Pi}^e,{\bf 0})$ is positive definite and we apply Theorem \ref{conditions-stability-Hessian}.
\end{proof}

\begin{rem}
The analysis of the above inequalities leads us to the following sufficient conditions for the stability  of the equilibrium point \eqref{equilibrium-spacecraft}.
\begin{itemize}
\item [(i)] $I_2>I_3>I_1$, and $k_1>k_3$ and $\omega_T^2>2(k_2-k_3)$;

\item [(ii)] $I_3>I_2>I_1$, and $k_1>k_3$ and $2(k_2-k_1)<\omega_T^2<2(k_2-k_3)$;

\item [(iii)] $I_3>I_1>I_2$, and $k_1>k_3$ and $\omega_T^2<2(k_2-k_1)$;

\item [(iv)] $I_2>I_1>I_3$, and $k_1<k_3$ and $\omega_T^2>2(k_2-k_1)$;

\item [(v)] $I_1>I_2>I_3$, and $k_1<k_3$ and $2(k_2-k_3)<\omega_T^2<2(k_2-k_1)$;

\item [(vi)] $I_1>I_3>I_2$, and $k_1<k_3$ and $\omega_T^2<2(k_2-k_3)$. $\Box$
\end{itemize}

\end{rem}

According to \cite{wang-xu} the radius of the stationary orbit $R_S$ satisfies the relation
\begin{equation}\label{Rs}
R_S^5-\frac{GM}{\omega_T^2}\left(R_S^2-\frac{3}{2}a_e^2C_{20}-9a_e^2C_{22}\right)=0,
\end{equation}
where $M$ is the mass of the asteroid, $G$ is the Gravitational constant, $a_e$ is the mean radius of the asteroid, and $C_{20}$, $C_{22}$ are the harmonic coefficients generated by the gravity field of the asteroid. Also, we have
\begin{equation}
k_1=\frac{3GMa_e^2C_{22}}{R_S^5},\,\,k_2=\frac{3GMa_e^2C_{20}}{2R_S^5},\,\, k_3=\frac{3GM}{2R_S^3}-\frac{3GMa_e^2}{4R_S^5}\left(5C_{20}+34C_{22}\right).
\end{equation}

If we use the above expressions of the coefficient $k_1,k_2$ and $k_3$, then the conditions (56a), (56b), and (56c) from the paper \cite{wang-xu} coincide with the inequalities from Theorem \ref{stability-spacecraft}. In the paper \cite{wang-xu} is used a modified energy-Casimir method in order to find the cited conditions. This method work with the eigenvalues of a $12\times 12$ matrix.

For a specific asteroid the coefficients $M,a_e,\omega_T,C_{20}$ and $C_{22}$ are specified.  The harmonic coefficients $C_{20}$ and $C_{22}$ are calculated by the formulas:
$$C_{20}=-\frac{1}{2Ma_e^2}(2I_w^A-I_u^A-I_v^A),\,\,C_{22}=\frac{1}{4Ma_e^2}(I_v^A-I_u^A),$$
where $I_u^A,I_v^A$ and $I_w^A$ are the principal moments of the asteroid which are calculated in the mass center of the asteroid.
The numbers of the positive solutions of the equation \eqref{Rs} can be characterized by the parameters of the asteroid. The equation can have two positive solutions, one positive solution or no positive solutions.
If we fix a radius of the stationary orbit around a specific asteroid, then we obtain fixed values of the coefficients $k_1,k_2$ and $k_3$.
\medskip

{\bf The stability of a spacecraft moving around the asteroid 4769 Castalia.}
According to \cite{scheers} the asteroid 4769 Castalia has the following physical data: $M=1.4091\cdot 10^{12}\,kg$, $a_e=543.1\,m$, $\omega_T=4.2882\cdot 10^{-4}\,s^{-1}$, $C_{20}=-7.257\cdot 10^{-2}$, $C_{22}=2.984\cdot 10^{-2}$. The universal constant has the value $G=6.67384\cdot 10^{-11}\,m^3kg^{-1}s^{-1}$. The equation \eqref{Rs} has two positive solutions: $R_{S1}=219.31\,m$ and $R_{S2}=778.39\,m$. Because $R_{S1}<a_e$ we can not have a stationary orbit with the radius $R_{S1}$. For a stationary orbit with the radius $R_{S2}$ we have $k_1<k_3$, $\omega_T^2>2(k_2-k_1)$ and,  consequently a sufficient condition for the stability of \eqref{equilibrium-spacecraft} is that the inertia moments of the spacecraft satisfies the inequalities $I_2>I_1>I_3$.

\section{Stability of an underwater vehicle}
Following \cite{leonard-1996}, the dynamics for a six degree-of-freedom vehicle modeled as a neutrally buoyant, submerged rigid body in an infinitely large volume of irrotational, incompressible, inviscid fluid that is at rest at infinity is described by the system
\begin{equation}\label{underwater-1}
\left\{%
\begin{array}{ll}
\dot{\boldsymbol \Pi}={\boldsymbol \Pi}\times {\boldsymbol \Omega}+{\bf Q}\times {\bf v}-mgl{\boldsymbol \Gamma}\times {\bf r} \\
\dot{\bf Q}={\bf Q}\times {\boldsymbol \Omega} \\
\dot{\boldsymbol \Gamma }={\boldsymbol \Gamma}\times {\boldsymbol \Omega},
\end{array}%
\right.
\end{equation}
where ${\boldsymbol \Pi}$ is the angular impulse, ${\bf Q}$ is the linear impulse, ${\boldsymbol \Gamma}$ is the direction of gravity, $l{\bf r}$ is the vector from center of buoyancy to the center of gravity (with $l\geq 0$ and ${\bf r}$ an unit vector), $m$ is the mass of the vehicle, $g$ is gravitational acceleration, ${\boldsymbol \Omega}$ and ${\bf v}$ are the angular and translational velocity of the vehicle. In a body-fixed frame with the origin in the centre of buoyancy the relationship between $({\boldsymbol \Pi},{\bf Q})$ and  $({\boldsymbol \Omega},{\bf v})$ is given by
\begin{equation}\label{relation-between}
\left(%
\begin{array}{c}
{\boldsymbol \Pi} \\
{\bf Q}
\end{array}%
\right)=\left(%
\begin{array}{cc}
J & D \\
D^T & M
\end{array}%
\right)\left(%
\begin{array}{c}
{\boldsymbol \Omega} \\
{\bf v}
\end{array}%
\right)
,
\end{equation}
where $J$ is the matrix that is the sum of the body inertia matrix plus the added inertia matrix associated with the potential flow model of the fluid, $M$ is the sum of the mass matrix for the body alone, and $D$ accounts for the cross terms.
According to \cite{leonard-automatica} we have
\begin{equation}\label{matrices-underwater}
M=m I_3+\Theta_{11}^f,\,\,J=J_b+\Theta_{11}^f,\,\,D=ml\widehat{\bf r}+\Theta_{21}^f.
\end{equation}
We denote by $I_3$ the identity matrix and by $J_b$ is the inertia matrix of the vehicle. The matrix
\begin{equation}\label{Theta-f}
\Theta^f=
\left(%
\begin{array}{cc}
\Theta_{11}^f &(\Theta_{21}^f)^T \\
\Theta_{21}^f & \Theta_{22}^f
\end{array}%
\right)
\end{equation}
is symmetric and it is determined by the configuration of the vehicle and the density of the fluid.
The relationship between $({\boldsymbol \Omega},{\bf v})^T$ and $({\boldsymbol \Pi},{\bf Q})^T$ is given by
\begin{equation}\label{relation-between}
\left(%
\begin{array}{c}
{\boldsymbol \Omega} \\
{\bf v}
\end{array}%
\right)=
\left(%
\begin{array}{cc}
A & B^T \\
B & C
\end{array}%
\right)
\left(%
\begin{array}{c}
{\boldsymbol \Pi} \\
{\bf Q}
\end{array}%
\right),\,\,\left(%
\begin{array}{cc}
A & B^T \\
B & C
\end{array}%
\right)=\left(%
\begin{array}{cc}
J & D \\
D^T & M
\end{array}%
\right)^{-1}.
\end{equation}
The matrix
$
\left(%
\begin{array}{cc}
A & B^T \\
B & C
\end{array}%
\right)$
 is symmetric and positive definite and consequently the matrices $A$ and $C$ are symmetric and positive definite.

The system \eqref{underwater-1} has the Hamilton-Poisson form (see \cite{leonard-automatica})
\begin{equation}\label{underwater-1-Poisson}
\dot{\bf u}=\Lambda({\bf u})\nabla H({\bf u}),
\end{equation}
where ${\bf u}=({\boldsymbol \Pi},{\bf Q},{\boldsymbol \Gamma})$,
$\Lambda({\bf u})=\left(%
\begin{array}{ccc}
\widehat{{\boldsymbol \Pi}} & \widehat{\bf Q} & \widehat{{\boldsymbol \Gamma}} \\
\widehat{\bf Q} & O_3 &O_3 \\
 \widehat{{\boldsymbol \Gamma}} & O_3 &O_3
\end{array}%
\right),$ and the Hamiltonian function is given by
$$H({\bf u})=\frac{1}{2}(<{\boldsymbol \Pi},A{\boldsymbol \Pi}>+2<{\boldsymbol \Pi},B^T{\bf Q}>+<{\bf Q},C{\bf Q}>-2mgl<{\boldsymbol \Gamma},{\bf r}>).
$$
In this case the Casimir functions are 
$C_{11}({\bf u})=||{\bf Q}||^2$, $C_{12}({\bf u})=<{\bf Q},{\boldsymbol \Gamma}>$ and $C_{22}({\bf u})=||{\boldsymbol \Gamma}||^2$ and we denote by ${\bf C}:\R^{9}\rightarrow \R^3$ the vectorial Casimir function.

We are interested in a generic equilibrium point ${\bf u}^e=({\boldsymbol \Pi}^e,{\bf Q}^e,{\boldsymbol \Gamma}^e)$ which satisfy
\begin{equation}\label{condition-equilibria}
{\bf Q}^e\neq {\bf 0},\,\,{\boldsymbol \Gamma}^e\neq {\bf 0},\,\,<{\bf Q}^e,{\boldsymbol \Gamma}^e>=0.
\end{equation}
This equilibria has the following properties:

\noindent (i) The equilibrium point has no spin. Because ${\bf u}^e$ is a generic equilibrium point we have
$${\boldsymbol \Omega}^e=A{\boldsymbol \Pi}^e+B^T{\bf Q}^e=\frac{\partial H}{\partial {\boldsymbol \Pi}}({\bf u}^e)={\bf 0}.$$

\noindent (ii) The translational velocity is ${\bf v}^e=B{\boldsymbol \Pi}^e+C{\bf Q}^e$.

\noindent (iii) The vector ${\bf r}$ is located in the plane generated by the vectors ${\bf Q}^e$ and ${\boldsymbol \Gamma}^e$ (we have $<{\bf Q}^e\times {\boldsymbol \Gamma}^e, {\bf r}>=0$).

\begin{rem}
In the paper \cite{birtea-comanescu} is presented a stability study for a nongeneric equilibria of the system \eqref{underwater-1} which is situated to singular symplectic leaves that are not characterized as a preimage o a regular value of the Casimir functions.
\end{rem}

For an equilibrium point ${\bf u}^e$ of our type we construct a coordinate chart on the regular symplectic leaf  ${\bf C}^{-1}({\bf C}({\bf u}^e))$ around the equilibrium point using the same scheme as in the previous section. The open subset ${\bf C}^+_{{\bf u}^e}\subset {\bf C}^{-1}({\bf C}({\bf u}^e))$ defined by:
$${\bf C}^+_{{\bf u}^e}=\{{\bf u}=({\bf x},{\bf y}_1,{\bf y}_2)\in {\bf C}^{-1}({\bf C}({\bf u}^e))\,|\,{\bf y}_1\times {\bf y}_2 \,\text{and}\,{\bf y}_1^e\times {\bf y}_2^e\,\text{have the same orientation}\}$$
contains the equilibrium point ${\bf u}^e$. Between $SE(3)$ and  ${\bf C}^+_{{\bf u}_e}$ we have the homeomorphism 
$$\mathcal{F}_{{\bf u}^e}(\boldsymbol{\Pi},R):=(\boldsymbol{\Pi}, R{\bf Q}^e, R\boldsymbol{\Gamma}^e)$$
with the property  $\mathcal{F}_{{\bf u}^e}(\boldsymbol{\Pi}^e,I_3):={\bf u}^e$. By using the notations of the previous section we have the following result.

\begin{prop}\label{chart-reduced}
There exists an open subset $U\subset {\bf C}^+_{{\bf u}^e}\subset  {\bf C}^{-1}({\bf C}({\bf u}^e))$ such that ${\bf u}^e\in U$  and $U$ with function $(\mathcal{F}_{{\bf u}^e}\circ \tilde{\bf P}\circ\mathcal{G})^{-1}:U\rightarrow V\subset \R^6$ is a coordinate chart.
\end{prop}

The Hamiltonian function written in the coordinates of the chart defined in the above proposition is:
$$H_{{\bf u}^e}(\boldsymbol{\Pi},{\bf p})=H(\boldsymbol{\Pi},R^{(\sqrt{1-||{\bf p}||^2},{\bf p})}{\bf Q}^e, R^{(\sqrt{1-||{\bf p}||^2},{\bf p})}\boldsymbol{\Gamma}^e ).$$

Analogously with the proofs of Theorem \ref{chart-spacecraft} and Theorem \ref{chart-spacecraft-Hess} we obtain the following  stability results.

\begin{thm}
 Let ${\bf u}^e$ be a generic equilibrium point satisfying conditions \eqref{condition-equilibria}.
If the point $(\boldsymbol{\Pi}^e,{\bf 0})$ is  a strict local extremum of $H_{{\bf u}^e}$, then the equilibrium point ${\bf u}^e$ is Lyapunov stable.

\end{thm}

\begin{thm}\label{conditions-stability-Hessian-reduced}
 Let ${\bf u}^e$ be a generic equilibrium point satisfying conditions \eqref{condition-equilibria}.
If $(\boldsymbol{\Pi}^e,{\bf 0})$ is a stationary point for $H_{{\bf u}^e}$ and the Hessian matrix $\text{Hess}\,H_{{\bf u}^e}(\boldsymbol{\Pi}^e,{\bf 0})$ is positive or negative definite, then the equilibrium point is Lyapunov stable.
\end{thm}

\subsection{Stability for an ellipsoidal bottom-heavy underwater vehicle}

In this section we suppose that the vehicle can be approximated by an ellipsoid. The origin of the body-fixed frame is located in the center of the buoyancy and we set the axes to be the principal axes of the displaced fluid. In this case the matrix $\Theta^f$ (see \eqref{Theta-f}) is a diagonal matrix. Suppose that the vector ${\bf r}$ is along the third axis; more precisely ${\bf r}=(0,0,1)$. In the papers \cite{leonard-automatica}, \cite{leonard-1996}, \cite{leonard-1996-1}, \cite{leonard-marsden} is supposed that the principal axes of displaced fluid coincide with the principal axes of the vehicle. In this paper we consider a more general situation, we suppose that the third axis is a principal axis of inertia for the vehicle but the first and second axes of the vehicle-fixed frame may not be principal axes of inertia for the vehicle. Using our hypotheses and the relation \eqref{matrices-underwater} we deduce:
\begin{equation}
M=\left(%
\begin{array}{ccc}
m_1 & 0 & 0 \\
0 & m_2 & 0\\
0 & 0 & m_{3}
\end{array}%
\right),
J=\left(%
\begin{array}{ccc}
I_{11} & I_{12} & 0 \\
I_{12} & I_{22} & 0\\
0 & 0 & I_{3}
\end{array}%
\right)
,
D=ml\widehat{\bf r}=\left(%
\begin{array}{ccc}
0 & -ml & 0 \\
ml & 0 & 0\\
0 & 0 & 0
\end{array}%
\right).
\end{equation}
\begin{figure}[h!]
  \caption{Underwater with noncoincident center of gravity and buoyancy.}
  \centering
    \includegraphics[width=0.8\textwidth]{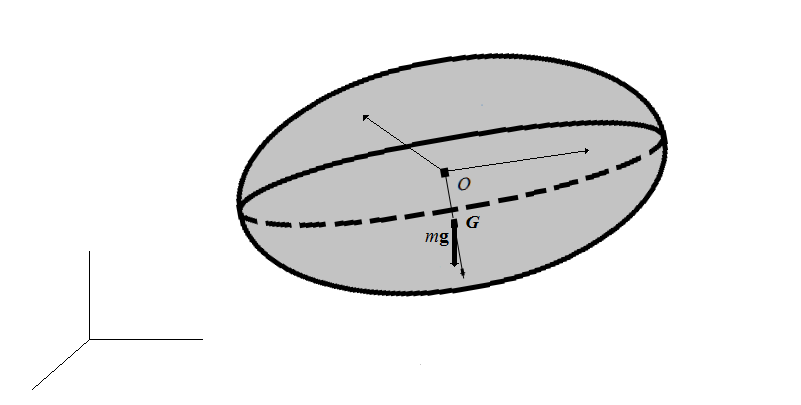}
\end{figure}
By physical reasons we have the inequalities:
\begin{equation}\label{inequality-physics-1}
m_1>0,\,\,m_2>0,\,\,m_3>0,\,\,I_{11}>0,\,\,I_{22}>0,\,\,I_{11}I_{22}-I_{12}^2>0,\,\,\,I_3>0.
\end{equation}
In this case we have
$$A=\frac{1}{k}\left(%
\begin{array}{ccc}
m_2(m_1I_{22}-m^2l^2) & -m_1m_2I_{12} & 0 \\
-m_1m_2I_{12} & m_1(m_2I_{11}-m^2l^2) & 0\\
0 & 0 & \frac{k}{I_3}
\end{array}%
\right),$$
$$B=\frac{1}{k}\left(%
\begin{array}{ccc}
m_1m_2I_{12} & -(m_2I_{11}-m^2l^2)ml & 0  \\
(m_1I_{22}-m^2l^2)ml & -m_1m_2I_{12} & 0 \\
 0 & 0 & 0
\end{array}%
\right).
$$
$$C=\frac{1}{k}\left(%
\begin{array}{ccc}
m_2(I_{11}I_{22}-I_{12}^2)-m^2l^2I_{22} & m^2l^2I_{12} & 0 \\
m^2l^2I_{12} & m_1(I_{11}I_{22}-I_{12}^2)-m^2l^2I_{11} & 0\\
0 & 0 & \frac{k}{m_3}
\end{array}%
\right),$$
$$k=m_1m_2(I_{11}I_{22}-I_{12}^2)-m^2l^2(m_1I_{22}+m_2I_{11})+m^4l^4.$$
By calculus we observe that the matrix $
\left(%
\begin{array}{cc}
A & B^T \\
B & C
\end{array}%
\right)$ is positive definite if and only if the following inequalities are satisfied:
\begin{equation}\label{inequality-physics-1}
m_1I_{22}>m^2l^2,\,\,m_2I_{11}>m^2l^2,\,\,k>0.
\end{equation}

We study a generic equilibrium point ${\bf u}^e$ with the coordinates:
\begin{equation}\label{equilibrium-underwater-particular-10}
\Pi_1^e=-\frac{ml}{m_2}Q_2^e,\,\,\Pi_2^e=\Pi_3^e=0,
Q_1^e=0,\,\,Q_2^e\in\R^*,\,\,Q_3^e=0,
\Gamma_1^e=\Gamma_2^e=0,\,\,\Gamma_3^e=1.
\end{equation}
This equilibrium point correspond to a constant translation with no spin along a horizontal direction of a bottom-heavy underwater vehicle.
By using the coordinate chart defined above, we have:
\small{$${\bf Q}=((2p_1p_2-2p_3\sqrt{1-(p_1^2+p_2^2+p_3^2)})Q_2^e, (1-2p_1^2-2p_3^2)Q_2^e,(2p_2p_3+2p_1\sqrt{1-(p_1^2+p_2^2+p_3^2)})Q_2^e)$$}
$$\boldsymbol{\Gamma}=(2p_1p_3+2p_2\sqrt{1-(p_1^2+p_2^2+p_3^2)}, 2p_2p_3-2p_1\sqrt{1-(p_1^2+p_2^2+p_3^2)}, 1-2p_1^2-2p_2^2).$$

We can present the main result of this section.
\begin{thm}
If we have
\begin{equation}
Q_2^e\neq 0,\,\,l>0,\,\,mgl>\left(\frac{1}{m_2}-\frac{1}{m_3}\right)(Q_2^e)^2,\,\,m_2>m_1,
\end{equation}
then the equilibrium point \eqref{equilibrium-underwater-particular-10} is Lyapunov stable.
\end{thm}

\begin{proof}
By direct calculus we observe that the point $(-\frac{ml}{m_2}Q_2^e,0,0,0,0,0)$  is a critical point for the function $H_{{\bf u}^e}$. The Hessian matrix calculated in the equilibrium point has the components:
\begin{align*}
h_{11} & =\frac{m_2}{k}(m_1I_{22}-m^2l^2),\,h_{12}=-\frac{m_1m_2I_{12}}{k},\,h_{13}=h_{14}=h_{15}=0,\,h_{16}=-\frac{2m_2mlI_{12}}{k}Q_2^e, \\
h_{22} & = \frac{m_1}{k}(m_2I_{11}-m^2l^2),\,h_{23}=h_{24}=h_{25}=0,\,h_{26}=\frac{2ml}{k}(m_2I_{11}-m^2l^2)Q_2^e, \\
h_{33} & =  \frac{1}{I_3}, h_{34}=h_{35}=h_{36}=0, \\
h_{44} & =4(mgl+(\frac{1}{m_3}-\frac{1}{m_2})(Q_2^e)^2),\,h_{35}=h_{36}=0, \\
h_{55} & =4mgl,\,h_{56}=0,\\
h_{66} & = -\frac{4(Q_2^e)^2}{km_2}(k-m_2^2(J_{11}J_{22}-J_{12}^2)+m_2m^2l^2J_{22}).
\end{align*}
The determinant of the Hessian matrix is
$$\frac{64mgl}{kI_{3}}(m_2-m_1)\left(mgl+(\frac{1}{m_3}-\frac{1}{m_2})(Q_2^e)^2\right)(Q_2^e)^2.$$
By using the hypotheses of this Theorem and the inequalities \eqref{inequality-physics-1} we observe that the Hessian matrix is positive definite.  Applying the Theorem \ref{conditions-stability-Hessian-reduced} we obtain the announced result.
\end{proof}

We remark that the conditions for the stability of an equilibrium point of type \eqref{equilibrium-underwater-particular-10} do not depend by the position of the first and second axes of inertia of the vehicle in the perpendicular plane on the third axis of the vehicle-fixed frame. The conditions for the stability of the above theorem have been obtained in Theorem 2 of the paper \cite{leonard-automatica} for the particular case when the principal axes of inertia of the displaced fluid and the principal axes of inertia of the underwater vehicle are coincident. In the paper \cite{leonard-1996-1} is used the energy-Casimir method to prove the stability of an equilibrium point.

\end{document}